\documentclass[a4paper, onecolumn, 11pt, accepted=2026-11-08]{quantumarticle}
\pdfoutput=1
\usepackage{graphicx}
\usepackage[utf8]{inputenc}
\usepackage[top=1.2in, bottom=1.2in, left=1.2in, right=1.2in]{geometry}
\usepackage[qm]{qcircuit}

\usepackage{amssymb,amsmath,url,etoolbox,multirow,multicol,enumerate}
\usepackage{graphicx,epstopdf,epsfig}
\usepackage{tikz-cd} 
\usepackage{inputenc}
\usepackage{amsmath, mathrsfs,  amssymb,amsbsy,amsfonts,amsthm,latexsym,amsopn,amstext, amsxtra,euscript,amscd}
\usepackage{physics}

\usepackage{makecell}
\usepackage{mathtools}
\usepackage{amsthm}
\usepackage{amsmath}
\usepackage{amssymb}
\usepackage{tikz}

\usetikzlibrary{arrows}

\newtheorem{theo}{Theorem}
\newtheorem{prop}[theo]{Proposition}
\newtheorem{cor}[theo]{Corollary}
\newtheorem{lem}[theo]{Lemma}
\theoremstyle{definition}
\newtheorem{defi}[theo]{Definition}

\newtheorem{exa}[theo]{Example}

\newtheorem{rem}[theo]{Remark}

\newtheorem{prob}{Problem}

\numberwithin{theo}{section}

\newcommand{\rmv}[1]{}

\newcommand{\U}{{\mathbb U}}
\newcommand{\F}{{\mathbb F}}

\newcommand{\Z}{{\mathbb Z}}

\newcommand{\cC}{{\mathcal C}}

\newcommand{\cJ}{{\mathcal J}}

\newcommand{\cP}{{\mathcal P}}

\newcommand{\ba}{{\mathbf a}}
\newcommand{\bb}{{\mathbf b}}
\newcommand{\bc}{{\mathbf c}}
\newcommand{\bd}{{\mathbf d}}
\newcommand{\be}{{\mathbf e}}

\newcommand{\bv}{{\mathbf v}}

\newcommand{\bx}{{\mathbf x}}
\newcommand{\by}{{\mathbf y}}

\AtBeginDocument{}
\def\a{\alpha}

\renewcommand{\Res}{{\mathrm{\sf  Res}}}

\newcommand{\sfC}{{\mathrm{\sf C}}}
\newcommand{\sfS}{{\mathrm{\sf S}}}
\newcommand{\sfZ}{{\mathrm{\sf Z}}}

\newcommand{\Sp}{\mbox{\sf Sp}}

\renewcommand{\dim}{{\sf dim}}

\newcommand{\rk}{{\sf rank}}

\newcommand{\rs}{{\sf rs}}

\newcommand{\inners}[2]{\mbox{$\langle{\,{#1}\,}|\,{#2}\,\rangle_{\sf s}$}}

\newcommand{\GL}{\mathrm{\sf GL}}
\newcommand{\Sym}{\mathrm{\sf Sym}}

\def\be{\begin{equation}}
\def\ee{\end{equation}}
\def\matrix#1{\begin{bmatrix}#1\end{bmatrix}} 

\newcounter{alp}
\newcounter{ara}
\newcounter{rom}

\newenvironment{alphalist}{\begin{list}{(\alph{alp})\hfill}{\usecounter{alp}
     \topsep0.4ex \labelwidth.6cm \leftmargin.6cm \labelsep0cm
     \rightmargin0cm \parsep0ex \itemsep0ex}}{\end{list}}
\newenvironment{arabiclist}{\begin{list}{(\arabic{ara})\hfill}{\usecounter{ara}
     \topsep0.4ex \labelwidth.6cm \leftmargin.6cm \labelsep0cm
     \rightmargin0cm \parsep0ex \itemsep0ex}}{\end{list}}
\newif\ifcomment
\commenttrue

\usepackage{hyperref}
\hypersetup{
  colorlinks   = true, 
  urlcolor     = blue, 
  linkcolor    = blue, 
  citecolor   = red 
}

\title{Climbing the Clifford Hierarchy}
\date{} 

\author[1]{Luca Bastioni}
\email{lbastioni@usf.edu}
\author[1]{Akash Bhople}
\email{akashbhople@usf.edu}
\author[2]{Samuel Glandon}
\email{nhg311@mocs.utc.edu}
\author[1]{Tefjol Pllaha}
\email{tpllaha@usf.edu}
\author[3]{Madison Stewart}
\email{mstewart314@utexas.edu}
\author[1]{Phillip Waitkevich}
\email{phillipwaitkevich@usf.edu}

\affil[1]{University of South Florida, Tampa, FL 33620}

\affil[2]{University of Tennessee at Chattanooga, Chattanooga, TN 37403}
\affil[3]{University of Texas at Austin, Austin, TX 78712}

\begin{document}

\maketitle
\begin{abstract}
The Clifford Hierarchy has been a central topic in quantum computation due to its strong connections with fault-tolerant quantum computation, magic state distillation, and more. 
Nevertheless, only sections of the hierarchy are fully understood, such as diagonal gates and third level gates. 
The diagonal part of the hierarchy can be climbed by taking square roots and adding controls. Similarly, square roots of Pauli gates (first level) are Clifford gates (climb to the second level). 
Based on this theme, we study gates whose square roots climb to the next level. 
In this paper, we primarily focus on Hermitian gates, and in particular, we fully characterize Hermitian Clifford gates whose square roots climb to the third level.
\end{abstract}


\section{Introduction and Motivation}
The Clifford hierarchy, introduced in 1999 by Gottesman and Chuang~\cite{Gottesman-nature99}, is a mathematical structure that plays a central role in fault-tolerant quantum-computation.
It is defined as the union of infinitely many nested {\em levels} $\cC^{(1)} \subset \cC^{(2)}\subset \cdots$, where the first level is the Pauli group $\cP$ and the $k$th level is comprised of all those gates that conjugate Pauli gates to the $(k-1)$st level.
Essentially, the hierarchy captures those quantum gates that can be implemented fault-tolerantly via the gate teleportation model and the level of the hierarchy quantifies the difficulty of said implementation.

The second level of the Clifford hierarchy coincides with the Clifford group.
Although the Clifford group can be efficiently simulated by a classical computer~\cite{gottesman98,AG04stabilziersimulation}, it is not sufficient for universal quantum computation.
However, adding any non-Clifford gate to the Clifford group does yield a universal gate set and usual non-Clifford contenders are taken from the third level of the hierarchy.

In light of its importance, the Clifford hierarchy has been the focus of many works over the years~\cite{ZXI08,Beigi-qic10,CGK17,PRTC20} as well as recent works~\cite{Anderson2024,HRT24,HRT25,Silva25,doi:10.1073/pnas.2515667122}.
Despite all this attention, the full structure of the hierarchy remains unknown and only isolated fragments are known. 
The best understood gates in the hierarchy (other than Pauli and Clifford gates, of course) are diagonal gates within each level.
This understanding stems from the fact that they form a group~\cite{ZXI08}, and this paved the way for a full classification~\cite{CGK17}.
One of its distinguishing features is that one can ``climb'' the diagonal part of the hierarchy by starting with a diagonal gate from the level below and either adding a control or taking a square root~\cite{HLCclimbing}.
Along this line, the recent work~\cite{xu2026controlledjumpcliffordhierarchy} shows that controlled Clifford gates climb the hierarchy. Concretely, they show that if $C$ is a Clifford gate then its controlled version is in $\cC^{(k+2)}\setminus \cC^{(k+1)}$ if $C^{2^k}$ is a Pauli gate and $k$ is the smallest such number.

In this paper, we explore whether 
square roots and their combination with controls can be generalized to explore the entire hierarchy.
Let $U$ be any Hermitian gate, and let us define
\begin{equation*}\label{e-maindef}
\widehat{U} = {\sf exp}\left(i\frac{\pi}{4}U\right) = \frac{I + iU}{\sqrt{2}}.
\end{equation*}
One can easily verify that $\widehat{U}$ is unitary if and only if $U$ is Hermitian and $\widehat{U}^2 = iU$, thus constituting a\footnote{Everything discussed in this paper also holds for ${\sf exp}(-i\pi U/4) = (I-iU)/\sqrt{2}$.} square root of $U$.
As mentioned, if $P$ is a diagonal Pauli gate (that is, a first level gate) then $\widehat{P}$ is a Clifford gate (that is, a second level gate), and in general, if $P\in \cC^{(k)}$ is a Hermitian\footnote{We point out here that the Hermitian condition can be dropped for diagonal gates. We merely keep it for uniformity with non-diagonal gates.} diagonal gate then $\widehat{P}\in \cC^{(k+1)}$.
Additionally, it is well-known that if $P$ is a Hermitian Pauli gate then $\widehat{P}$ is a Clifford gate, and in fact, $\{\widehat{P}\mid P\in \cP = \cC^{(1)}\}$ generate the Clifford group $\cC^{(2)}$~\cite{RCKP20, PVT21}.
Unlike the diagonal gates, this behavior stops at the base of hierarchy, since, for instance, the Hadamard gate $H$ does not satisfy $\widehat{H}\in \cC^{(3)}$; see also Example~\ref{exa-hadamard}. 
This motivates the following.
\begin{prob}\label{Prob}
Let $U\in \cC^{(k)}$ be a Hermitian gate. Determine under what conditions $\widehat{U}$ ``climbs'' to the level above, that is $\widehat{U}\in \cC^{(k+1)}$. 
\end{prob}

To determine whether or not $\widehat{U}\in \cC^{(k+1)}$, one must examine the action of $\widehat{U}$ on a Pauli gate $P$ (recall that $U$ is Hermitian and thus $U^\dagger = U)$, namely,
\begin{equation}\label{e-mainconj}
\widehat{U}P\widehat{U}^\dagger = \frac{1}{2}\left(P + iUP - iPU + UPU\right).
\end{equation}
In this paper, we examine~\eqref{e-mainconj} and determine cases and conditions for $\widehat{U}P\widehat{U}^\dagger \in \cC^{(k)}$, which in turn yields cases and conditions for $\widehat{U}\in \cC^{(k+1)}$.
Among other things, towards settling Problem~\ref{Prob}, we accomplish the following results:
\begin{arabiclist}
\item In Theorem~\ref{T-counter}, we stablish necessary conditions for Hermitian gates $U\in \cC^{(k)}$ such that $\widehat{U}\in \cC^{(k+1)}$.
\item Show that CNOT gates climb the hierarchy. Essentially, this is a consequence of the fact that controlled-$Z$ gates climb the hierarchy; see Remark~\ref{R-CXclimbs}. In Theorem~\ref{T-CNOT}, we also lay out relevant conjugations.
\item In Theorems~\ref{T-gen-climb} and~\ref{T-converse}, we fully characterize Hermitian Clifford gates that climb to the third level.
\item In Section~\ref{Sec5}, we combine square roots and controls to climb to the fourth level.
\end{arabiclist}

While the primary reason of this work is mathematical curiosity, we are cautiously optimistic that beyond a better understanding of Clifford hierarchy this work will have important applications in fault-tolerant quantum computing. 
Indeed,~\eqref{e-maindef} falls under the general umbrella of $\varphi$-rotations of Hermitian gates $U$ which are given by $U_\varphi = \cos(\varphi)I + i \sin(\varphi)U$.
For instance, if $U$ is a Hermitian Pauli gate, then Pauli rotations $U_\varphi$ are the building blocks of the Trotter-Suzuki approximation~\cite{trotter1959product,suzuki1976generalized} and play a central role in quantum simulation~\cite{lloyd1996universal}. 
While $\pi/8$-rotations of Hermitian Pauli gates are sufficient for universal computation, the state-of-the-art implementation relies on the resource-heavy magic state distillation~\cite{magicstate}.
Unlocking $\pi/4$-rotations of more complicated gates (such as Hermitian Clifford gates) may yield shorter circuits.

The content of the paper is organized as follows. In Section~\ref{Sec2}, we set up the notation and list well-known facts about the hierarchy that will be referenced throughout the paper. 
In Section~\ref{Sec3}, we start by laying out some examples and non-examples, which in turn shed light to some necessary conditions for climbing the hierarchy.
In Section~\ref{Sec4}, we focus on Hermitian Clifford gates, and determine necessary and sufficient conditions for climbing to the third level. 
For this, we leverage connections with symplectic geometry and the powerful machinery therein.
In Section~\ref{Sec5}, we focus on Hermitian Clifford gates (as classified in Section~\ref{Sec4}) and show that square roots and controls can be combined to lift them to the fourth level.
We end the paper with some conclusions and future directions in Section~\ref{Sec6}.


\section{Preliminaries}\label{Sec2}
The {\em Pauli gates} consist of the identity gate $I_2$, the bit-flip gate $X$, the phase-flip gate $Z$, and the bit-phase-flip gate $Y = iXZ$. Namely,
\begin{equation*}
I_2 = \matrix{1&0\\0&1},\, X = \matrix{0&1\\1&0},\,Z = \matrix{1&0\\0&-1}, \, Y = \matrix{0&-i\\i&0}. 
\end{equation*}
The {\em Pauli group on $1$ qubit} is the multiplicative group generated by the Pauli gates.
Namely,
\begin{equation*}
\cP_1 = \{\pm I_2,\pm iI_2, \pm X, \pm iX, \pm Z, \pm iZ,\pm Y,\pm iY\}.
\end{equation*}
Note that a Pauli gate $P\in \cP_1$ is of the form
\begin{equation*}
P = i^cX^aZ^b, \quad a,b\in \F_2, c\in \Z_4.
\end{equation*}
Thus, we may denote a Pauli gate as $P(a,b;c)$. If $c=0$, the Pauli gate will be denoted $P(a,b)$.
The {\em Pauli group on $n$ qubits} is defined as $\cP_n = \cP_1^{\otimes n}$. That is, $P\in \cP_n$ is of the form $P = P_1\otimes P_2\otimes \cdots \otimes P_n$ where $P_j\in \cP_1$. If we write $P_j = P(a_j,b_j;c_j) = i^{c_j}X^{a_j}Z^{b_j}$, we have that
\begin{align*}
P & = P_1\otimes P_2\otimes \cdots \otimes P_n\\
& = i^{c_1}X^{a_1}Z^{b_1}\otimes i^{c_2}X^{a_2}Z^{b_2}\otimes \cdots \otimes i^{c_n}X^{a_n}Z^{b_n}\\
& = i^{c_1+\cdots+c_n}(X^{a_1}\otimes X^{a_2}\otimes \cdots \otimes X^{a_n})(Z^{b_1}\otimes Z^{b_2}\otimes \cdots \otimes Z^{b_n})\\
&\equiv P(\ba,\bb;c)
\end{align*}
where $\ba = (a_1,a_2,\ldots,a_n), \bb = (b_1,b_2,\ldots,b_n)\in \F_2^n$, $c = c_1+c_2+\cdots+c_n\in \Z_4$.
We will abuse notation and call $P\in \cP_n$ a Pauli gate as well.
Note that $E = P(\ba,\bb;c)\in \cP_n$ is Hermitian if and only if $c = \ba\bb^\top$. 
We will use this notation to distinguish between generic Pauli gates that satisfy $P^2 = \pm I$ and Hermitian Pauli gates that satisfy $E^2 = I$.
Finally, we will denote $X_i$ and $Z_i$ the Pauli basis with $X$ and $Z$ in qubit $i$, respectively, and identity otherwise.

It is well-known that Pauli gates either commute or anti-commute
\be\label{e-comm}
P(\ba,\bb)P(\bc,\bd) = (-1)^{\inners{(\ba,\bb)}{(\bc,\bd)}}P(\bc,\bd)P(\ba,\bb),
\ee 
where $\inners{(\ba,\bb)}{(\bc,\bd)} = \ba\bd^\top + \bb\bc^\top$ denotes the {\em symplectic inner product}.

The {\em Clifford hierarchy on $n$ qubits} is defined as a union of infinitely many levels
\begin{equation*}
\cC_n = \bigcup_{k = 1}^\infty \cC_n^{(k)},
\end{equation*}
where the first level is the Pauli group, that is, $\cC^{(1)}_n = \cP_n$ and higher levels are defined as $\cC_n^{(k)} = \{U\in \mathbb{U}(2^n)\mid U\cP_nU^\dagger \subset \cC_n^{(k-1)}\}.$ 
We will drop the subscript $n$ when it is irrelevant or clear from the context.
The second level $\cC^{(2)}_n$ is called the {\em Clifford group on $n$ qubits}.
We will refer to a gate from the second level as a {\em Clifford gate}. 
Let us denote
\begin{equation*}
S = \matrix{1&0\\0&i} \text{ and } H = \frac{1}{\sqrt{2}}\matrix{1&1\\1&-1}
\end{equation*}
the phase gate and the Hadamard gate respectively. 
It is well-known that these gates generate the Clifford group on one qubit.

Next, we list some well-known results on the Clifford hierarchy.
\begin{theo}[\cite{ZXI08}]\label{T-closed}
\begin{alphalist}
\item The levels of the Clifford hierarchy are closed under multiplication by Clifford gates. That is, if $U\in \cC^{(k)}$ and $C_1,C_2\in \cC^{(2)}$ then $C_1UC_2\in \cC^{(k)}$.
\item Diagonal gates in each level form a group.
\end{alphalist}
\end{theo}
For $U\in\U(2^j)$, let $\sfC^{(i)}(U) \in \U(2^{i+j})$ denote a controlled-$U$ gate where the first $i$ qubits are control qubits and the last $j$ qubits are targets. 
Note that this arrangement is done without loss of generality since the target qubits can be changed with SWAP gates (that is, Clifford gates) and by Theorem~\ref{T-closed}, this does not affect the level of hierarchy.

The following result is well-known, which we list here for reference.
\begin{theo}\label{T-diagonal}
We have that $\sfC^{(k)}(Z) \in \cC^{(k+1)}$. More generally, we have that $\sfC^{(i)}(\!\sqrt[2^j]{Z})\in \cC^{(i+j+1)}$~\cite{ZXI08}, and these gates generate the group of diagonal gates in the respective level~\cite{CGK17}.
As a corollary, since $\sfC^{(k)}(X)$ and $\sfC^{(k)}(Z)$ only differ by Clifford gates (conjugation by the Hadamard gate), they belong to the same level of the hierarchy.
\end{theo}
\begin{cor}\label{C-diag}
Since $S,S^\dagger$ are square roots of $Z$, we have that $\sfC^{(k)}(S),\sfC^{(k)}(S^\dagger)\in \cC^{(k+2)}$.
\end{cor}
\begin{cor}\label{C-CP}
Let $P\in \cC^{(1)}$ be a Pauli gate. Then $\sfC(P)\in \cC^{(2)}$.
\end{cor}
\section{Some examples and non-examples}\label{Sec3}
We start this section with the formal definition of the main operation considered in this paper, along with some basic properties.
\begin{defi}
Let $U$ be any Hermitian gate. Define $\widehat{U} = {\sf exp}\left(i\pi U/4\right) = (I + iU)/\sqrt{2}$.
\end{defi}
\begin{lem}
Let $U$ be a Hermitian gate. Then, $\widehat{U}$ satisfies the following:
\begin{alphalist}
\item $\widehat{U}$ is unitary, $\widehat{U}^2 = iU$, and $\widehat{U}$ has order 8.
\item $\widehat{U}^\dagger = {\sf exp}(-i\pi U/4) = (I-iU)/\sqrt{2}$.
\end{alphalist}
\end{lem}
Before we describe classes of gates $U\in \cC^{(k)}$ for which $\widehat{U} \in \cC^{(k+1)}$, let us start with an interesting non-example, which in turn will motivate a necessary condition (see Theorem~\ref{T-counter}) for $\widehat{U}$ to climb the Clifford hierarchy.
\begin{exa}\label{exa-hadamard}
Consider the Hadamard gate $H\in \cC_1^{(2)}$. Then, one can easily verify that $\widehat{H}\notin \cC_1^{(3)}$.
As we will see below, this is due to the fact that $HXH^\dagger = Z$ and $X,Z$ anti-commute.
Perhaps the more interesting fact is that $\widehat{H}$ does not belong to the Clifford hierarchy at all.
Indeed, recall that $H = (X+Z)/\sqrt{2}$. Thus $\widehat{H} = (\sqrt{2}I+iX+iZ)/2$.
Then, one computes
$\widehat{H}X\widehat{H}^\dag = -\widehat{H} Y$. If $\widehat{H}$ belonged to some finite level of the hierarchy, then by definition, $\widehat{H}X\widehat{H}^\dagger$ would belong to the level below. This is a contradiction, since $\widehat{H}X\widehat{H}^\dagger$ differs from $\widehat{H}$ only by $-Y$, and by Theorem~\ref{T-closed}(a) they would belong to the same level.
\end{exa}

In general, we have the following.
\begin{lem}\label{L-cliff}
Let $E,E'\in \cP$ be two anti-commuting Hermitian Pauli gates. Then $C = (E+E')/\sqrt{2}$ is a Clifford gate.
\end{lem}
\begin{proof}
Since $E,E'$ are Hermitian, we have that $C$ is indeed unitary.
Additionally, if $P\in \cP$ is any Pauli gate, it is straightforward to verify that
\be\label{e-cliff}
CPC^\dagger = 
\begin{cases}
P, & \text{if }EP = PE\text{ and } E'P=PE',\\
PEE', & \text{if }EP = PE\text{ and } E'P=-PE',\\
-PEE', & \text{if }EP = -PE\text{ and } E'P=PE',\\
-P, & \text{if }EP = -PE\text{ and } E'P=-PE',
\end{cases}
\ee
and this concludes the proof.
\end{proof}
\begin{lem}\label{L-dagger}
Let $U$ be a Hermitian gate such that $U\in \cC^{(k)}$. If $UEU = E'$ for some anticommuting Hermitian Pauli gates $E,E'$ then $\widehat{U}\in \cC^{(k+1)}$ if and only if $\widehat{U}^\dagger\in \cC^{(k+1)}$.
\end{lem}
\begin{proof}
Note that $U$ being Hermitian and $UEU = E'$ implies $UE'U = E$. Thus $U(EE')U = (UEU)(UE'U) = E'E = - EE'$. 
Let $E''$ be the corresponding Hermitian Pauli gate, for which we have $E''\widehat{U}E'' = \widehat{U}^\dagger$, and the statement follows by Theorem~\ref{T-closed}(a).
\end{proof}
\begin{theo}\label{T-counter}
Let $U\in \cC^{(k)}$ be a Hermitian gate such that $UE = E'U$ for some anti-commuting Hermitian Pauli gates $E,E'\in \cP$. 
Then $\widehat{U}\notin \cC^{(k+1)}$.
\end{theo}
\begin{proof}
We have that 
\begin{align*}
\widehat{U}E\widehat{U}^\dagger &= \frac{1}{2}\left(E + iUE - iEU + UEU\right)\\
& = \frac{1}{2}E\left(I + iEUE - iU + EE'\right)\\
& = \frac{1}{2}E\left(I - iUEE' - iU + EE'\right)\\
& = E \cdot \frac{I-iU}{\sqrt{2}} \cdot \frac{I+EE'}{\sqrt{2}}.
\end{align*}
By Lemma~\ref{L-cliff} we have that $(I+EE')/\sqrt{2} = E(E+E')/\sqrt{2}\in \cC^{(2)}$ and by Lemma~\ref{L-dagger} we have that $\widehat{U}^\dagger = (I-iU)/\sqrt{2}\in \cC^{(k)}$. 
Thus, Theorem~\ref{T-closed} implies that $\widehat{U}E\widehat{U}^\dagger$ and $\widehat{U}$ belong to the same (minimal) level. Finally, $\widehat{U}$ cannot be in $\cC^{(k+1)}$.
\end{proof}
\begin{exa}\label{exa-hadamard1}
Let $E,E'\in\cP$ be two anti-commuting Hermitian Pauli gates so that $C = (E+E')/\sqrt{2}\in \cC^{(2)}$. 
By~\eqref{e-cliff}, we see that $C$ satisfies the conditions of Theorem~\ref{T-counter}, and thus $\widehat{C}\notin\cC^{(3)}$. Note that this observation also addresses Example~\ref{exa-hadamard} since $H = (X+Z)/\sqrt{2}$ and $XZ = -ZX$.
\end{exa}

Let us now return to some examples for which taking the square root does climb the Clifford hierarchy.
\begin{exa}[Lifting CNOT]
Let $C = \sfC(X)$. Since $C$ fixes $X_2$ and $Z_1$ under conjugation, the same is true for $\widehat{C}$. Next, let $R = HS^\dagger H\in \cC^{(2)}$. 
One straightforwardly verifies that
\begin{align*}
\widehat{C}X_1\widehat{C}^\dagger &= X\otimes R\cdot \sfC(X),\\
\widehat{C}Z_2\widehat{C}^\dagger & = S^{\dag}\otimes Z\cdot \sfC(X).
\end{align*}
Combining Corollary~\ref{C-diag} with Theorem~\ref{T-closed}, we have that $\sfC(R)\in \cC^{(3)}$. 
Of course, we have that $X\otimes R, \,\sfC(X)\in \cC^{(2)}$. 
Thus $\widehat{C}$ conjugates the Pauli basis into Cliffords. 
Since $\cC^{(2)}$ is a group, this is sufficient to show that $\widehat{C}\in \cC^{(3)}$.
\end{exa}
\begin{exa}[Lifting SWAP]\label{exa-swap}
Let $U$ denote the SWAP gate in two qubits. We then have
\begin{align*}
\widehat{U}X_1\widehat{U}^\dagger &= U\cdot H\otimes H\cdot \sfC(Z)\cdot S \otimes S \cdot H\otimes H\cdot X_1,\\
\widehat{U}X_2\widehat{U}^\dagger &= U\cdot H\otimes H\cdot \sfC(Z)\cdot S \otimes S \cdot H\otimes H\cdot X_2,\\
\widehat{U}Z_1\widehat{U}^\dagger &= U\cdot \sfC(Z)\cdot S\otimes S \cdot Z_1,\\
\widehat{U}Z_2\widehat{U}^\dagger &= U\cdot \sfC(Z)\cdot S\otimes S \cdot Z_2.
\end{align*}
We conclude that $\widehat{U}\in \cC^{(3)}$.
\end{exa}
As mentioned in Theorem~\ref{T-diagonal}, if we let $U = \sfC^{(k)}(Z)\in \cC^{(k+1)}$ we indeed have that $\widehat{U}\in \cC^{(k+2)}$ does climb the hierarchy.
This is due to the facts that 
\begin{equation*}
\widehat{Z} = \frac{1+i}{\sqrt{2}}S^\dagger \text{ and } \widehat{U} = \frac{1+i}{\sqrt{2}}\sfC^{(k)}(S^\dagger).
\end{equation*}
As a consequence, since $X = HZH$, due to Theorem~\ref{T-closed}, the same is true for $U = \sfC^{(k)}(X)$.
\begin{rem}[$\widehat{\sfC^{(k)}(X)}$ climbs the hierarchy]\label{R-CXclimbs}
Consider again  the gate 
\begin{equation}\label{e-Rgate}
R = HS^\dagger H = \frac{1}{2}\matrix{1-i&1+i\\1+i&1-i}\in \cC^{(2)}.
\end{equation}
Then, by Theorem~\ref{T-diagonal}, we have that $\sfC^{(k)}(R)\in \cC^{(k+2)}$. Next, note that 
\begin{equation*}
\widehat{\sfC^{(k)}(X)} = \frac{1+i}{\sqrt{2}}\sfC^{(k)}(R),
\end{equation*}
which then again implies $\widehat{\sfC^{(k)}(X)}\in \cC^{(k+2)}$.
\end{rem}

Although Remark~\ref{R-CXclimbs} already establishes $\widehat{\sfC^{(k)}(X)}$ climbs the hierarchy, below we present the exact conjugation rules that verify the result and that can also be used for gate synthesis.
\begin{theo}\label{T-CNOT}
Let $U = \sfC^{(k)}(X)$. Then, if we denote $\sfS_{1,i}$ to be the gate that swaps qubits 1 and $i$ (in $k+1$ qubits), the following conjugation rules hold.
\begin{align*}
\widehat{U}X_1\widehat{U}^\dagger & = X\otimes \sfC^{(k-1)}(R)\cdot U \equiv M_1,\\
\widehat{U}X_i\widehat{U}^\dagger & = \sfS_{1,i} M_1 \sfS_{1,i}^\dagger, \text{ for } i = 1,\ldots,k,\\
\widehat{U}X_{k+1}\widehat{U}^\dagger & =X_{k+1},\\
\widehat{U}Z_i\widehat{U}^\dagger & = Z_i, \text{ for } i = 1,\ldots,k,\\
\widehat{U}Z_{k+1}\widehat{U}^\dagger & =\sfC^{(k-1)}(S)\otimes Z\cdot U.
\end{align*}
\end{theo}
\begin{proof}
Since $U$ fixes $X_{k+1}$ and $Z_i,\,i=1,\ldots,k$ under conjugation, so does $\widehat{U}$.
To see the action of $\widehat{U}$ on $X_1$, one can proceed in many different ways, including induction or action on basis states. However, we present here a simple argument based on block matrices.
Indeed, if we denote $U_1 = \sfC^{(k-1)}(X)$, we have
\begin{align*}
\widehat{U}X_1\widehat{U}^\dagger & = \frac{1}{2}\matrix{(1+i)I&\\&I+iU_1}\cdot \matrix{&I\\I&}\cdot \matrix{(1-i)I&\\&I-iU_1}\\
& = \frac{1}{2}\matrix{&(1+i)(I-iU_1)\\(1-i)(I+iU_1)&}.
\end{align*}
Recall the gate $R = HS^\dagger H$ from~\eqref{e-Rgate}. Directly by construction, we have $\sfC^{(k-1)}(R)=(1-i)(I+iU_1)/2$.
Thus, we have
\begin{align*}
\widehat{U}X_1\widehat{U}^\dagger & = \matrix{&\sfC^{(k-1)}(R^\dagger)\\\sfC^{(k-1)}(R)&}\\
& = \matrix{&\sfC^{(k-1)}(R)\\\sfC^{(k-1)}(R)&}\cdot\matrix{I&\\&\sfC^{(k-1)}(X)} & \text{since } R X = R^\dagger\\
& =  X\otimes \sfC^{(k-1)}(R)\cdot \sfC^{(k)}(X).
\end{align*} 

Next, for $i = 1,\ldots, k$, note first that $X_i = \sfS_{1,i}X_1\sfS_{1,i}^\dagger$. It is also easy to see that $\sfS_{1,i}$ commutes with $U$ (since $U$ affects only the last qubit), and thus it also commutes with $\widehat{U}$. This simple observation implies the corresponding conjugation rules.

Finally, a similar argument as for $X_1$ works for $Z_{k+1}$, and for this reason we omit the details. This concludes the proof.
\end{proof}
We end this section with a discussion of non-Hermitian gates, which as we will see, behave entirely differently from Hermitian gates.

\begin{exa}\label{exa-sqrtT}
Let $E$ be a Hermitian Pauli gate. Then, as discussed $C = \widehat{E}$ is a {\it non-Hermitian} Clifford gate. 
Note that $U = \cos(\pi/8)I+i\sin(\pi/8)E = {\sf exp}(i\pi E/8)$ satisfies $U^2 = \widehat{E}$.
Then, for any Pauli gate $P$, we have that 
\begin{equation}\label{e-nonHerm}
UPU^\dagger = 
\begin{cases}
P, & \text{if } PE = EP,\\
\widehat{E}P, & \text{if } PE = -EP.
\end{cases}
\end{equation}
Since $\widehat{E}\in \cC^{(2)}$, we conclude that $UPU^\dagger\in \cC^{(2)}$ and thus $U\in \cC^{(3)}$. 
A similar argument shows that the gate $\cos(\pi/2^k)I+i\sin(\pi/2^k)E\in \cC^{(k)}$.
\end{exa}
\begin{exa}\label{exa-sqrtT1}
Let $\{E_1,\ldots, E_\ell\}$ be a set of {\it independent}\footnote{A set of Pauli gates is called independent if their defining binary vectors are independent.} commuting Hermitian Pauli gates. Then $C = \prod_{j=1}^\ell\widehat{E_j}$ is a non-Hermitian Clifford gate.
As in Example~\ref{exa-sqrtT}, let us denote $U_j={\sf exp}(i\pi E_j/8)$. Then $U = \prod_{j=1}^\ell U_j$ satisfies $U^2 = C$.
Applying~\eqref{e-nonHerm} consecutively, we have that the action of $U$ on a Pauli gate $P$ is given by 
\begin{equation*}
UPU^\dagger = \left(\prod_{j\in \cJ}\widehat{E_j}\right)P,
\end{equation*}
where $\cJ$ is the collection of those $E_j$'s that anti-commute with $P$. Thus $U\in \cC^{(3)}$.
\end{exa}
\section{Lifting Clifford gates}\label{Sec4}
In this section, we will leverage the symplectic representation of the Clifford group.
For details, we refer the reader to~\cite{PRTC20,RCKP20}.
Let $C \in \cC_n^{(2)}$. Then, by definition, $C$ conjugates a Hermitian Pauli to another Hermitian Pauli. Let us focus on the generators $X_i = E({\bf e}_i,{\bf0}),Z_i = E({\bf0},{\bf e}_i)$ of the Pauli group, and let 
\be\label{e-conj}
CX_iC^\dagger = \pm E(\ba_i,\bb_i), \quad CZ_iC^\dagger = \pm E(\bc_i,\bd_i),
\ee 
Let $F_C$ denote the $2n\times 2n$ matrix formed by the vectors $(\ba_i,\bb_i),(\bc_i,\bd_i)$. 
Then by construction, for any Hermitian Pauli gate $E(\ba,\bb)$,~\eqref{e-conj} reads as
\begin{equation}\label{e-conj1}
CE(\ba,\bb)C^\dagger = \pm E((\ba,\bb)F_C).
\end{equation}
It is well-known that $F_C$ is a {\em symplectic matrix}, which means that it satisfies
\begin{equation*}
F_C\Omega_n F_C^\top = \Omega_n, \quad \Omega_n = \matrix{0_n&I_n\\I_n&0_n}.
\end{equation*}
The collection of all such matrices forms the {\em binary symplectic group} $\Sp(2n)$.
Note that symplectic matrices are precisely those matrices that preserve the symplectic inner product.
\begin{rem}
\begin{alphalist}
\item The map $C \mapsto \Sp(2n)$ is a surjective homomorphism of groups with kernel $\cP_n$, and thus, we have $\cC_n^{(2)}/\cP_n \cong \Sp(2n)$.
This means that any for Clifford gate, up to a Pauli gate, there exists a unique representative symplectic matrix.
For details and specific mapping, we refer the reader to~\cite{PVT21}.
\item If $C\in \cC_n^{(2)}$ is Hermitian, we have that $C = C^\dagger = C^{-1}$ and thus $F_C = F_C^{-1}$.
Symplectic matrices $F$ for which $F^2=I$ are called {\em symplectic involutions}.
Thus, when considering $\widehat{C}$, we must restrict to symplectic involutions.
\end{alphalist}
\end{rem}
Let $\bv\in \F_2^{2n}$ and let $E(\bv)$ be its associated Hermitian Pauli gate.
It is easy to verify that 
\be\label{e-Cv}
C_\bv = \frac{I_N \pm iE(\bv)}{\sqrt{2}}
\ee 
is a Clifford gate. The associated symplectic matrix is $F_\bv = I_{2n}+\Omega_n\bv^\top\bv$. In symplectic geometry, $F_\bv$ is called a {\em symplectic transvection}, and for this reason, we will refer to such $C_\bv$ as a {\em Clifford transvection}.
The following is well-known.
\begin{theo}[\cite{Omeara}]\label{T-transvection}
The symplectic group is generated by symplectic transvections. As a consequence, up to a Pauli gate, every Clifford gate is a product of Clifford transvections.
\end{theo}
Before proceeding with specific classes of Clifford gates, let us first address the necessary condition of Theorem~\ref{T-counter}.
Let $C$ be Clifford gate. Then, as a necessary condition for $\widehat{C}$ to climb to the third level, we must have that $CEC^\dagger$ (which already is a Pauli gate) must commute with $E$ for all Hermitian Pauli gates $E$.
If we let $E = \pm E(\bv),\,\bv\in \F_2^{2n}$ and $F_C$ the symplectic matrix that corresponds to $C$, then according to equations~\eqref{e-comm} and~\eqref{e-conj1}, we must have that 
\be\label{e-hyp}
\inners{\bv F_C}{\bv} = 0 ,\, \text{for all } \bv\in \F_2^{2n}.
\ee
Symplectic matrices that satisfy~\eqref{e-hyp} are called {\em hyperbolic}.

We have proved the following.
\begin{theo}\label{T-cliff}
If $C$ is a Clifford gate such that $F_C$ is not hyperbolic, then $\widehat{C}\notin\cC^{(3)}$.
\end{theo}
The transvection decomposition of a symplectic matrix $F$ (as in Theorem~\ref{T-transvection}) is governed by the {\em residue space} defined as
\begin{equation*}
\Res(F) = \rs(I+F) = \{\bv+\bv F\mid \bv\in \F_2^{2n}\}.
\end{equation*}
Here $\rs(\bullet)$ is used to denote the row space of a matrix. Due to Theorem~\ref{T-cliff}, we restrict to Clifford gates for which $F_C$ {\it is} hyperbolic. 
For hyperbolic involutions, we have the following.
\begin{theo}[\mbox{\cite[2.1.8]{Omeara}}]\label{T-res}
Let $F$ be a hyperbolic involution and let $r = \dim(\Res(F))$. Then, there exists a basis $\{\bv_1,\ldots,\bv_r\}$ of $\Res(F)$ and $\bv\in \Res(F)$ such that $F = F_{\bv_1}\cdots F_{\bv_r}F_\bv$.
\end{theo}
\begin{rem}
Due to Theorem~\ref{T-cliff}, we focus on hyperbolic involutions $F_C$. In this case, the quotient $\F_2^{2n}/\Res(F_C)$ is a non-degenerate quadratic space and as such must have even dimension~\cite{taylor1992geometry}. 
As a consequence, $r = \dim(\Res(F_C))$ is even. 
This is particularly straightforward to see if $C$ is a diagonal Hermitian Clifford gate which is characterized by a binary symmetric matrix with all-zero diagonal; see also Remark~\ref{R-sym-symp}.
If this dimension is $r=0$ then $C$ is a Pauli gate. Thus, the first non-trivial case is $r=2$. 
As we will see in Theorem~\ref{T-gen-climb} and~\ref{T-converse}, $r=2$ is the unique non-trivial residue dimension for which $\widehat{C}$ is in the third level.
\end{rem}
The main goal of this section is to prove the following necessary and sufficient conditions. The sufficient condition will be proved in Theorem~\ref{T-gen-climb} (see also Theorem~\ref{T-diag-climb} for diagonal gates) and the necessary condition will be proved in Theorem~\ref{T-converse}.

\begin{theo}
Let $C\in \cC^{(2)}\setminus\cC^{(1)}$ be Hermitian. Then $\widehat{C}\in \cC^{(3)}$ if and only if $\dim(\Res(F_C))=2$.
\end{theo}
\subsection{The case of diagonal Clifford gates}
We start by considering diagonal Clifford gates, which are well-known to be determined by quadratic forms; see~\cite[Appendix A]{RCKP20} for instance.
Let $C\in \cC_n^{(2)}$ be a diagonal Clifford gate. Then, there exists an $n\times n$ binary symmetric matrix $A\in \Sym(n)$ such that 
\begin{equation*}
C = [i^{\bx A \bx^\top}]_{\bx\in \F_2^n}.
\end{equation*}
The equation holds up to phases and Pauli gates and the computation on the exponent is modulo 4. 
Accordingly, the respective symplectic matrix is
\be\label{e-sym-symp}
F_C = \matrix{I&A\\0&I}, \, A^\top = A.
\ee
Since we consider only Hermitian gates, we must have $\bx A \bx^\top = 0 \mod 2$ for all $\bx$ (so that $C$ has $\pm1$ entries). 
Since $A$ is symmetric and $\bx$ is binary, we have that 
\begin{equation*}
\bx A \bx^\top = \sum_{i=1}^nx_ia_{i,i} + 2\sum_{i<j}x_ix_ja_{i,j}
\end{equation*}
Thus, $\bx A\bx^\top = 0 \mod 2$ if and only if $A$ has all-zero diagonal. 
\begin{rem}\label{R-sym-symp}
Note that a symplectic matrix as in~\eqref{e-sym-symp} is hyperbolic if and only if $A$ has all-zero diagonal. 
Indeed, for $\bv = (\ba,\bb)$, we have that 
\begin{align*}
\inners{\bv F_C}{\bv} &= \inners{(\ba,\ba A+\bb)}{(\ba,\bb)}\\
& = \ba\bb^\top + \ba A\ba^\top+\bb\ba^\top\\
& = \ba A\ba^\top.
\end{align*}
Thus, a diagonal Clifford gate $C$ is Hermitian if and only if $F_C$ (as in~\eqref{e-perm-symp}) is hyperbolic, which coincides with the overall necessary condition of Theorem~\ref{T-cliff}.
\end{rem}
\begin{theo}\label{T-diag-climb}
Let $C\in \cC^{(2)}$ be a diagonal Clifford gate such that $F_C$ is a hyperbolic involution and $\dim(\Res(F_C))=2$. Then $\widehat{C}\in \cC^{(3)}$.
\end{theo}
\begin{proof}
Since $C$ is a diagonal gate, we have that $F_C$ is as in~\eqref{e-sym-symp}. Thus $\Res(F_C) = \{({\bf 0}, \bx)\mid \bx \in \rs(A)\}$.
Of course Clifford transvections of form $C_{({\bf 0}, \bx)}$ are diagonal gates.

Combining Theorem~\ref{T-res} with~\eqref{e-Cv}, we can write (up to a phase and a Pauli gate)
\begin{equation*}
C = \frac{I\pm iZ(\ba)}{\sqrt{2}}\cdot\frac{I\pm iZ(\bb)}{\sqrt{2}}\cdot\frac{I\pm iZ(\ba+\bb)}{\sqrt{2}}
\end{equation*}
for some diagonal Pauli gates $Z(\ba),Z(\bb)$ with $\ba,\bb\in \rs(A)$.
We may fix the signs without loss of generality, and we may assume that (after fixing the signs to $+,+,-$) 
\begin{equation*}\label{e-diagcliff}
C = \frac{1}{2}\Big(I+Z(\ba)+Z(\bb) - Z(\ba+\bb)\Big).
\end{equation*}
To show that $\widehat{C}\in \cC^{(3)}$, we must show that $\widehat{C}P\widehat{C}^\dagger \in \cC^{(2)}$ for all $P\in \cP$. 
But since $C$ is a diagonal gate, we only need to consider $P = X(\bc)$.
First, note that 
\begin{equation*}
CX(\bc)C^\dagger = \begin{cases}
X(\bc), & \text{ if } \ba\bc^\top = 0, \bb\bc^\top = 0,\\
Z(\ba)X(\bc), & \text{ if } \ba\bc^\top = 0, \bb\bc^\top = 1,\\
Z(\bb)X(\bc), & \text{ if } \ba\bc^\top = 1, \bb\bc^\top = 0,\\
-Z(\ba+\bb)X(\bc), & \text{ if } \ba\bc^\top = 1, \bb\bc^\top = 1.
\end{cases}
\end{equation*}
and 
\begin{equation*}
C\, X(\bc) - X(\bc)\,C = \begin{cases}
0, & \text{ if } \ba\bc^\top = 0, \bb\bc^\top = 0,\\
\Big(Z(\bb) - Z(\ba+\bb)\Big)X(\bc), & \text{ if } \ba\bc^\top = 0, \bb\bc^\top = 1,\\
\Big(Z(\ba) - Z(\ba+\bb)\Big)X(\bc), & \text{ if } \ba\bc^\top = 1, \bb\bc^\top = 0,\\
\Big(Z(\ba)+Z(\bb)\Big)X(\bc), & \text{ if } \ba\bc^\top = 1, \bb\bc^\top = 1.
\end{cases}
\end{equation*}
Combining all of this, we obtain
\begin{align*}
\widehat{C}\, X(\bc)\,\widehat{C}^\dagger & = \frac{1}{2}\Big(X(\bc) + iC\,X(\bc) - iX(\bc)\,C + C\,X(\bc)\,C^\dagger\Big)\\
& = \begin{cases}
X(\bc), & \text{ if } \ba\bc^\top = 0, \bb\bc^\top = 0,\\
\frac{1}{2}\Big(I + iZ(\bb)-iZ(\ba+\bb)+Z(\ba)\Big)\,X(\bc), & \text{ if } \ba\bc^\top = 0, \bb\bc^\top = 1,\\
\frac{1}{2}\Big(I + iZ(\ba)-iZ(\ba+\bb)+Z(\bb)\Big)\,X(\bc), & \text{ if } \ba\bc^\top = 1, \bb\bc^\top = 0,\\
\frac{1}{2}\Big(I + iZ(\ba)+iZ(\bb)-Z(\ba+\bb)\Big)\,X(\bc), & \text{ if } \ba\bc^\top = 1, \bb\bc^\top = 1.
\end{cases}
\end{align*}
The first case is obvious, and for the last three cases note that all those gates are a product of two Clifford transvections. This concludes the proof.
\end{proof}
\begin{exa}\label{nonexa}
One can verify that, for $n=3$ qubits, all the diagonal Clifford gates $C$ for which $F_C$ is hyperbolic satisfy $\widehat{C}\in \cC^{(3)}_3$.
For $n=4$ qubits, this is no longer the case.
Consider the matrix 
\begin{equation*}
A = \matrix{0&0&0&1\\0&0&1&0\\0&1&0&0\\1&0&0&0}.
\end{equation*}
Let $C\in \cC_4^{(2)}$ be the Clifford gate that corresponds to $F_C$ as in~\eqref{e-sym-symp}.
We have that $\dim(\Res(F_C)) = \dim(\rs(A))=4$. By Theorem~\ref{T-transvection}, we have that $C$ is (up to a Pauli gate and a phase) product of five Clifford transvections, and we can write it is a linear combination of $Z$-gates
\begin{equation*}
C = \frac{1}{4}\sum_{\ba\in \F_2^4}\alpha_\ba Z(\ba).
\end{equation*}
As it turns out, $\widehat{C}\notin\cC^{(3)}_4$ but $\widehat{C}\in\cC^{(4)}_4$. 
To see this, note first that since $\widehat{C}$ is diagonal, we need to only examine its action on $X$-gates.
Let us denote $\sfC\sfZ_{i,j}$ the control-$Z$ gate with qubit $i$ as control and qubit $j$ as target.
By construction, we have $C = {\sf CZ}_{1,4}\cdot {\sf CZ}_{2,3}$.
Then, the action of $\widehat{C}X_1\widehat{C}^\dagger$ is given by
\begin{align*}
X_1&\rightarrow X_1Z_4,\\
X_2&\rightarrow \frac{X_2}{2}\left(I+Z_3+Z_4-Z_3Z_4\right) = X_2 \cdot{\sf CZ}_{3,4},\\
X_3&\rightarrow \frac{X_3}{2}\left(I+Z_2+Z_4-Z_2Z_4\right) = X_3 \cdot{\sf CZ}_{2,4},\\
X_4&\rightarrow \frac{X_4Z_1Z_4}{2}\left(I+Z_2+Z_3-Z_2Z_3\right) = X_4Z_1Z_4 \cdot{\sf CZ}_{2,3}.
\end{align*}
Thus, $\widehat{C}X_1\widehat{C}^\dagger\in \cC^{(3)}_4$. 
A similar argument holds for other $X$-gates.
\end{exa}
\subsection{The case of general Clifford gates}
We now continue with the case of a general Hermitian Clifford gate. First, we will need the following result.
\begin{lem}\label{L-com}
Let $F$ be an involution. Then $\inners{\bx}{\by}=0$ for all $\bx,\by\in \Res(F)$.
\end{lem}
\begin{proof}
For $\bx,\by\in \Res(F)$, there exist $\bv_1,\bv_2$ such that $\bx = \bv_1+\bv_1F,\, \by = \bv_2+\bv_2F$.
Then, we have
\begin{align*}
\inners{\bx}{\by} & = \inners{\bv_1+\bv_1F}{\bv_2+\bv_2F}\\
& = (\bv_1+\bv_1F)\,\Omega \,(\bv_2+\bv_2F)^\top \\
& = \bv_1\Omega\bv_2^\top+\bv_1\Omega F^\top\bv_2^\top + \bv_1F\Omega\bv_2^\top + \bv_1F\Omega F^\top\bv_2^\top.
\end{align*}
Since $F$ is symplectic we have that $F\Omega F^\top = \Omega$ and since $F$ is an involution we have that $F\Omega=\Omega F^\top$, and the result follows.
\end{proof}
\begin{theo}\label{T-gen-climb}
Let $C \in \cC^{(2)}$ be a Clifford gate such that $F_C$ is hyperbolic and $\dim(\Res(F_C)) = 2$. Then $\widehat{C} \in \cC^{(3)}$.
\end{theo}
\begin{proof}
By assumption (as in the proof of Theorem~\ref{T-diag-climb}), we can write (up to a Pauli gate and a phase)
\be\label{e-Cv2}
C =  \frac{I\pm iE_1}{\sqrt{2}}\cdot\frac{I\pm iE_2}{\sqrt{2}}\cdot\frac{I\pm iE_1E_2}{\sqrt{2}}.
\ee
Since $C$ is assumed to be Hermitian, we have that corresponding symplectic matrix is an involution and thus $E_1,\, E_2$ commute due to Lemma~\ref{L-com}.
As such, there exists a Clifford gate $C_1$ (see~\cite[Algorithm 1]{RCKP20}, for instance) such that $C_1E_1C_1^\dag = Z_1, C_1E_2C_1^\dag = Z_2$ are diagonal Clifford gates.
Thus, (after perhaps fixing the signs in~\eqref{e-Cv2}) we have that 
\begin{align*}
C &= \frac{1}{2}\Big(I+E_1+E_2-E_1E_2\Big)    \\
& = C_1^\dagger\left[ \frac{1}{2}\Big(I+Z_1+Z_2-Z_1Z_2\Big) \right] C_1.
\end{align*} 
If we denote $D = (I+Z_1+Z_2-Z_1Z_2)/2$, we have that $ \widehat{C}= C_1^\dagger \widehat{D} C_1$, and the result follows by Theorem~\ref{T-closed} and Theorem~\ref{T-diag-climb}.
\end{proof}
\begin{rem}\label{R-diag}
The argument of Theorem~\ref{T-gen-climb} implies that we only need to consider Hermitian diagonal Clifford gates. 
Indeed, any Hermitian Clifford gate, as described above, can be diagonalized via some other Clifford gate, and the resulting diagonal gate governs the behavior of the original gate.
\end{rem}

Guided by the ideas of Example~\ref{nonexa}, we will show next that the conditions of Theorem~\ref{T-gen-climb} are also necessary. 
To accomplish this, we first need some preparation.
Let $C\in \cC_n^{(2)}$ be a Hermitian Clifford gate such that $\dim(\Res(F_C)) = r$.
As a generalization of~\eqref{e-diagcliff} (see also~\cite{PRTC20} for details), we can write
\begin{equation}\label{e-cliffsupport}
C = \frac{1}{\sqrt{2^r}}\sum_{E\in S}\a_EE,
\end{equation}
where $S$ is a subgroup of the Pauli group $\cP_n$. 
We will refer to~\eqref{e-cliffsupport} as the {\em Pauli expansion} of $C$.
Since $C$ is Hermitian, we must have that $\a_E = \pm 1$.
Additionally, since the set of Hermitian Pauli gates forms an orthonormal basis (with respect to the Hermitian inner product), the coefficients can be computed as 
\begin{equation*}
\frac{\a_E}{\sqrt{2^r}} = \frac{1}{2^n}{\sf Tr}(EC).
\end{equation*}
Here ${\sf Tr}(\bullet)$ is used to denote the trace of a matrix.
We have proved the following (recall that, except the identity, Pauli gates are traceless).
\begin{lem}\label{L-trace}
Let $C\in \cC_n^{(2)}$ be as in~\eqref{e-cliffsupport}. For any Hermitian Pauli gate $E$ we have 
\begin{equation*}
|{\sf Tr}(EC)|=\begin{cases} 2^{n-r/2},&\text{if } E\in S,\\0,&\text{if } E\notin S.\end{cases}  
\end{equation*}
In particular, the magnitudes of all the coefficients in the Pauli expansion of $C$ are the same.
\end{lem}

We are now ready to prove the converse of Theorem~\ref{T-gen-climb}.
\begin{theo}\label{T-converse}
Let $C\in \cC_n^{(2)}$ be a Hermitian Clifford gate such that $\dim(\Res(F_C)) = r > 2$. 
Then $\widehat{C}\notin \cC_n^{(3)}$.
\end{theo}
\begin{proof}
We will prove the result by showing that there exist Pauli gates $\tilde{E}$ such that $\widehat{C}\tilde{E}\widehat{C}^\dagger \notin \cC_n^{(2)}$. 
We may assume that $\tilde{E} \neq C\tilde{E}C^\dagger$, since if we have equality for all Pauli gates then $C$ must be the identity.
We will show that 
\begin{equation*}
\widehat{C}\tilde{E}\widehat{C}^\dagger = \frac{1}{2}\left(\tilde{E}+iC\tilde{E}-i\tilde{E}C+C\tilde{E}C\right)\notin \cC_n^{(2)},
\end{equation*}
by showing that it violates Lemma~\ref{L-trace}. 
First, since $\tilde{E}^2 = I$, we have that
\begin{equation*}
{\sf Tr}(\tilde{E}\widehat{C}\tilde{E}\widehat{C}^\dagger) = \frac{1}{2}{\sf Tr}(I) + \frac{i}{2}{\sf Tr}(\tilde{E}C\tilde{E}) - \frac{i}{2}{\sf Tr}(C) + \frac{1}{2}{\sf Tr}(\tilde{E}C\tilde{E}C).
\end{equation*}
By the properties of the trace, we have that ${\sf Tr}(\tilde{E}C\tilde{E}) = {\sf Tr}(\tilde{E}^2C) = {\sf Tr}(C)$, and the two middle terms above cancel out.
Additionally, since $\tilde{E}\neq C\tilde{E}C$, we have that the product $\tilde{E}C\tilde{E}C^\dagger$ is non-identity Pauli gate and the last term is 0. 
Finally, we have that ${\sf Tr}(\tilde{E}\widehat{C}\tilde{E}\widehat{C}^\dagger) = 2^{n-1}$.

Next, let $\bar{E}$ be a Hermitian Pauli gate that anticommutes with $\tilde{E}$.
Again, we have
\begin{equation*}
{\sf Tr}(\bar{E}\widehat{C}\tilde{E}\widehat{C}^\dagger) = \frac{1}{2}{\sf Tr}(\bar{E}\tilde{E}) + \frac{i}{2}{\sf Tr}(\bar{E}C\tilde{E}) - \frac{i}{2}{\sf Tr}(\bar{E}\tilde{E}C) + \frac{1}{2}{\sf Tr}(\bar{E}C\tilde{E}C^\dagger).
\end{equation*}
Similarly, the first term is 0. For the last term, if $\bar{E} = C\tilde{E}C^\dagger$, then $\widehat{C}\tilde{E}\widehat{C}^\dagger\notin \cC^{(3)}_n$ due to Theorem~\ref{T-counter}. So we may assume $\bar{E} \neq C\tilde{E}C$ and the last term is also zero.
For the two middle terms this time we have ${\sf Tr}(\bar{E}C\tilde{E}) = {\sf Tr}(\tilde{E}\bar{E}C) = -{\sf Tr}(\bar{E}\tilde{E}C)$.
Thus ${\sf Tr}(\bar{E}\widehat{C}\tilde{E}\widehat{C}^\dagger) = -i{\sf Tr}(\bar{E}\tilde{E}C)$, which from Lemma~\ref{L-trace} equals either 0 (if $\bar{E}\tilde{E}\notin S$) or has magnitude $2^{n-r/2}$ (if $\bar{E}\tilde{E}\in S$). 
Under the assumption that $r>2$, we have that $2^{n-r/2} < 2^{n-1}$.
Since the coefficients of $\widehat{C}\tilde{E}\widehat{C}^\dagger$ in the Pauli expansion have different magnitudes, $\widehat{C}\tilde{E}\widehat{C}^\dagger$ cannot be a Clifford gate due to Lemma~\ref{L-trace}.
\end{proof}
\subsubsection{The case of permutation Clifford gates}
A similar analysis with the one for diagonal Clifford gates can be done for permutation Clifford gates.
Let $C\in \cC^{(2)}_n$ be a permutation Clifford gate. Then there exists an invertible matrix $B\in \GL(n)$ that realizes the qubit permutation $\ket{\bx}\mapsto \ket{\bx B}$.
The respective symplectic matrix is 
\begin{equation}\label{e-perm-symp}
F_C = \matrix{B&0\\0&B^{-\top}}, \text{ where } B^{-\top} = (B^{-1})^\top.
\end{equation}
Since a permutation Clifford gate has only real entries, it will be Hermitian if and only if it is symmetric. 
Note that $C^\top$ realizes the qubit permutation $\ket{\bx}\mapsto \ket{\bx B^{-1}}$. Thus, we have that $C = C^\top$ if and only if $B = B^{-1}$.
Similarly as with diagonal Clifford gates, we have the following.
Note that $F_C$ as in~\eqref{e-perm-symp} with $B = B^{-1}$ is also hyperbolic. 
Indeed, we have that
\begin{align*}
\inners{(\ba,\bb)F_C}{(\ba,\bb)} &= \ba B \bb^\top + \bb B^\top \ba^\top\\
&= \ba B \bb^\top + \ba B \bb^\top\\
& = 0,
\end{align*}
for all $(\ba,\bb)\in \F_2^{2n}$.

\begin{exa}\label{exa-cnots}
Let ${\sf CX}_{i,j}$ denote the control-$X$ gate with qubit $i$ as control and qubit $j$ as target.
Let $C_1 = {\sf CX}_{1,3},\, C_2 = {\sf CX}_{2,4}$, and put $C = C_1\cdot C_2$.
As seen previously, we have $\widehat{C}_1,\, \widehat{C}_2 \in \cC^{(3)}$. 
One can also verify that $\widehat{C}_1\cdot \widehat{C}_2 \in \cC^{(3)}$, while $\widehat{C}\notin\cC^{(3)}$.
We point out here that $\dim(\Res(F_C)) = 4$, so the latter statement also follows from Theorem~\ref{T-converse}. This demonstrates that lifting the product of (even commuting!) gates is not straightforward.
\end{exa}
\section{Lifting the third level}\label{Sec5}
In this section, we show that square roots and controls can be combined to climb to the fourth level of the Clifford hierarchy. Namely, we show that if $C$ is a Hermitian Clifford gate such that $\widehat{C}\in \cC^{(3)}$ (as characterized in Section~\ref{Sec4}) then $\sfC(\widehat{C}), \widehat{\sfC(C)}\in \cC^{(4)}$.
Note that this is an analogous result (albeit partial) with that of Theorem~\ref{T-diagonal}.

We start with an elementary auxiliary result that is part of a more general statement~\cite[Lemma 1]{fby6-xjbm}.
We only need this weaker version and include a proof for completeness.
\begin{lem}\label{L-CC}
Let $C\in \cC^{(2)}$ be a Clifford gate such that $C^2 = \pm I$. Then $\sfC(C)\in \cC^{(3)}$.
\end{lem}
\begin{proof}
A Pauli gate is either block-diagonal or block-off-diagonal and can be written as
\begin{equation}\label{e-genericPauli}
P = \matrix{P_1&\\&\pm P_1} \text{\, or \,} P = \matrix{&P_1\\\pm P_1},
\end{equation}
where $P_1$ is a generic Pauli gate.
In the former case, one easily verifies that 
\begin{equation}\label{e-genericPauli1}
\sfC(C) \cdot P\cdot[\sfC(C)]^\dagger = X_1\cdot \sfC(P_1)\cdot X_1\cdot \sfC(\pm CP_1C^\dagger).    
\end{equation}
Since $C\in \cC^{(2)}$, we have that $\pm CP_1C^\dagger \in \cC^{(1)}$. Since $\sfC(P_1), \sfC(\pm CP_1C^\dagger) \in \cC^{(2)}$ (due to Corollary~\ref{C-CP}) and since $\cC^{(2)}$ is a group, we have that $\sfC(C) P [\sfC(C)]^\dagger\in \cC^{(2)}$.

In the latter case, we have that
\begin{align*}
\sfC(C) \cdot P \cdot [\sfC(C)]^\dagger & = \matrix{&P_1C^\dagger\\\pm CP_1&}\\
& = \matrix{P_1&\\&I}\matrix{&I\\\pm CP_1C^\dagger&}\matrix{C^\dagger&\\&C^\dagger}\\
& = X_1\cdot \sfC(P_1)\cdot X_1\cdot \sfC(\pm CP_1C^\dagger)\cdot X_1\cdot I\otimes C^\dagger.
\end{align*}
In the second equation above we have used the assumption that $C^2 = \pm I$.
As before, every gate in the product above is either a Pauli gate or a Clifford gate and thus $\sfC(C) P [\sfC(C)]^\dagger\in \cC^{(2)}$.
\end{proof}
\begin{rem}\label{R-periodicity}
The condition $C^2 = \pm I$ (that is, $C$ is either Hermitian or anti-Hermitian) is not necessary.
A necessary condition for a controlled Clifford gate to be in the hierarchy is that its order is a power of 2~\cite{anderson2025controlledgatescliffordhierarchy}.
For instance, $HS$ has order 6 and $\sfC(HS)$ is not in the hierarchy.
This result was strengthened in~\cite[Theorem 3]{xu2026controlledjumpcliffordhierarchy} which explicitly shows that if $C^{2^k}\in \cP$ and $C^{2^t}\notin \cP$ for all $t<k$ then $\sfC(C)\in \cC^{(k+2)}\setminus \cC^{(k+1)}$.
Finally,~\cite[Lemma 1]{fby6-xjbm} shows that this is an ``if and only if'' statement for $k=1$.
\end{rem}
\begin{theo}\label{T-C4}
Let $C\in \cC^{(2)}$ be a Hermitian Clifford gate such that $\widehat{C}\in \cC^{(3)}$.
Then $U = \sfC(C)$ satisfies $\widehat{U}\in \cC^{(4)}$.
\end{theo}
\begin{proof}
As in Lemma~\ref{L-CC}, for a block-diagonal Pauli gate, we have that
\begin{align*}
\widehat{U}P\widehat{U}^\dagger  & = \matrix{\widehat{I}&\\&\widehat{C}}\matrix{P_1&\\&\pm P_1} \matrix{\widehat{I}^\dagger&\\&\widehat{C}^\dagger} \\
& = \matrix{P_1&\\&\pm \widehat{C}P_1\widehat{C}^\dagger}\\
& = X_1\cdot \sfC(P_1)\cdot X_1\cdot \sfC(\pm \widehat{C}P_1\widehat{C}^\dagger).
\end{align*}
By assumption, we have that $\widehat{C}\in \cC^{(3)}$, and thus $\pm \widehat{C}P_1\widehat{C}^\dagger \in \cC^{(2)}$.
In addition, $\pm \widehat{C}P_1\widehat{C}^\dagger$ has the same order as $\pm P_1$ (that is, 1, 2, or 4), and by Lemma~\ref{L-CC}, we have that $\sfC(\pm \widehat{C}P_1\widehat{C}^\dagger)\in \cC^{(3)}$.
Since $X_1, \sfC(P_1)\in \cC^{(2)}$ we conclude that $\widehat{U}P\widehat{U}^\dagger\in \cC^{(3)}$.

For block-off-diagonal Pauli gates, the argument reduces to the case of block-diagonal case since
\begin{align*}
\widehat{U}P\widehat{U}^\dagger  & = \matrix{\widehat{I}&\\&\widehat{C}}\matrix{&P_1&\\\pm P_1&} \matrix{\widehat{I}^\dagger&\\&\widehat{C}^\dagger} \\
& = \matrix{&P_1&\\\pm \widehat{C}P_1\widehat{C}^\dagger&}\\
& = \matrix{P_1&\\&\pm \widehat{C}P_1\widehat{C}^\dagger}\cdot X_1.
\end{align*}
This concludes the proof.
\end{proof}
In the next two examples, we illustrate the above by exploring the product of two Toffoli gates as well as a controlled-SWAP gate. 
\begin{exa}\label{exa-toff}
Let ${\sf T}_{i,j,k}$ denote the Toffoli gate with qubits $i,j$ as control and qubit $k$ as target.
\begin{alphalist}
\item Let us first consider $U_1 = {\sf T}_{i,j,k_1},\,U_2 = {\sf T}_{i,j,k_2}\in \cC^{(3)}$, that is, two Toffoli gates with the same control qubits $i,j$ but different target qubits $k_1\neq k_2$. 
We then have $U_1\cdot U_2 = {\sf C}^{(2)}(X_{k_1}X_{k_2})$.
Let us denote $C = {\sf C}(X_{k_1}X_{k_2})$.
Then $C$ satisfies the conditions of Theorem~\ref{T-C4}, namely, $C$ is a Clifford gate such that $C^2 = I$ and $\widehat{C}\in \cC^{(3)}$ due to Theorem~\ref{T-gen-climb} (since, for instance $\dim(\Res(F_C))=2$).
Then, for $U = U_1\cdot U_2 = {\sf C}(C)$, we have that $\widehat{U}\in \cC^{(4)}$ does climb to the fourth level due to Theorem~\ref{T-C4}.
\item If two Toffoli gates do not share any controls nor targets, then, as in Example~\ref{exa-cnots}, their product does not climb to the fourth level.
\item Next, we consider the case where two Toffoli gates share exactly one control. 
Let $U_1 = {\sf T}_{i,j_1,k_1},\,U_2 = {\sf T}_{i,j_2,k_2}$. We then have $U_1\cdot U_2 = {\sf C}({\sf CX}_{j_1,k_1}\cdot {\sf CX}_{j_2,k_2})$, where the first control is the shared control qubit $i$.
Again, the Clifford gate ${\sf CX}_{j_1,k_1}\cdot {\sf CX}_{j_2,k_2}$ does not climb to the third level due to Example~\ref{exa-cnots}, and so Theorem~\ref{T-C4} does not apply.
\end{alphalist}
\end{exa}
\begin{exa}\label{exa-swap2}
Consider the (Hermitian, Clifford) SWAP gate ${\sf S}$. By Lemma~\ref{L-CC}, we have ${\sf C(S)} \in \cC^{(3)}$ and by Example~\ref{exa-swap}, we have $\widehat{\sf S}\in \cC^{(3)}$.
Thus, the conditions of Theorem~\ref{T-C4} are met and for $U = {\sf C(S)}$ we have that $\widehat{U}$ does climb to the fourth level.
\end{exa}

We now continue with the case of controlled square roots.
First, we list some auxiliary elementary facts that will be used in the main result.
\begin{lem}\label{L-U2}
\begin{alphalist}
\item Let $U\in \cC^{(3)}$ be a diagonal gate. Then $U^2\in \cC^{(2)}$.
\item Let $C\in \cC^{(2)}$ be a diagonal gate. Then $\sfC(C)\in \cC^{(3)}$.
\end{alphalist}
\end{lem}
\begin{proof}
\begin{alphalist}
\item Note that for any $P\in \cP$ we have that $UPU^\dagger P^\dagger \in \cC^{(2)}$ is diagonal. If we denote this gate $D$, we have that $UPU^\dagger = DP$. 
But then
$U^2P[U^\dagger]^2 = UDPU^\dagger = DUPU^\dagger=D^2P\in \cP$ (since the square of a diagonal Clifford gate is a Pauli gate).
\item By assumption, we have that $C^2\in \cP$ and the result follows by Remark~\ref{R-periodicity}.
\end{alphalist}
\end{proof}
\begin{prop}\label{P-CU}
Let $U\in \cC^{(3)}$ be a diagonal gate. Then $\sfC(U)\in \cC^{(4)}$.
\end{prop}
\begin{proof}
As in Lemma~\ref{L-CC}, the argument splits into two cases. Let $P$ be a generic Pauli gate, which as in~\eqref{e-genericPauli}, is either block diagonal or block off-diagonal.
Again, as in~\eqref{e-genericPauli1}, we have that
\begin{equation*}
\sfC(U) \cdot \begin{bmatrix} P_1 & \\ & \pm P_1 \end{bmatrix} \cdot [\sfC(U)]^\dagger = X_1 \cdot \sfC(P_1) \cdot X_1 \cdot \sfC(\pm U P_1 U^\dagger).
\end{equation*}
Since $U\in \cC^{(3)}$ we have that $UP_1U^\dagger\in \cC^{(2)}$ and additionally, $(UP_1U^\dagger)^2 = UP_1^2U^\dagger = \pm I$. 
Thus, $\sfC(\pm U P_1 U^\dagger)\in \cC^{(3)}$ due to Lemma~\ref{L-CC}.

For the second case, we have
\begin{align*}
\sfC(U)\cdot \matrix{&P_1\\\pm P_1&}
\cdot [\sfC(U)]^\dagger & = \matrix{P_1U^\dagger P_1^\dagger&\\&U}\cdot  \matrix{&P_1\\\pm P_1&}\\
& = (I\otimes P_1U^\dagger P_1^\dagger)\cdot \sfC(P_1UP_1^\dagger U)\cdot  \matrix{&P_1\\\pm P_1&}.
\end{align*}
The fact that $I\otimes P_1U^\dagger P_1^\dagger\in \cC^{(3)}$ is clear. For the second gate, note that $P_1UP_1^\dagger U = P_1(UP_1^\dagger U^\dagger)U^2 \in \cC^{(2)}$ due to Lemma~\ref{L-U2}(a). Additionally, $P_1UP_1^\dagger U$ is diagonal, and thus $\sfC(P_1UP_1^\dagger U)\in \cC^{(3)}$ due to Lemma~\ref{L-U2}(b).
Finally, note that both gates are strictly diagonal, and thus their product is in $\cC^{(3)}$ due to Theorem~\ref{T-diagonal}.
\end{proof}
\begin{theo}\label{T-C4-1}
Let $C\in \cC^{(2)}$ be a Hermitian Clifford gate such that $\widehat{C}\in \cC^{(3)}$.
Then $\sfC(\widehat{C})\in \cC^{(4)}$.
\end{theo}
\begin{proof}
By Remark~\ref{R-diag}, $C$ can be diagonalized by another Clifford gate $Q$, that is, $QCQ^\dagger = D\in \cC^{(2)}$. 
This implies $\widehat{D} = Q\widehat{C}Q^\dagger$ and $\sfC(\widehat{C}) = (I\otimes Q^\dagger)\cdot \sfC(\widehat{D})\cdot (I\otimes Q)$. Now the result follows by Proposition~\ref{P-CU}.
\end{proof}

\section{Conclusions and future research}\label{Sec6}
The main goal of this paper is to study the Clifford hierarchy.
Motivated by the fact that the diagonal part of the Clifford hierarchy can be climbed by adding controls and taking square roots, we explore if and when the same can be done for the entire hierarchy.
Very fast one realizes that this is not always possible for every gate, and this opens up a series of avenues for exploration.

In this paper, we primarily focus on gates in the Clifford hierarchy whose square roots climbs one level higher.
In particular, we completely characterize Hermitian Clifford gates whose square roots climb to the third level.
Additionally, we show that, for the aforementioned Hermitian Clifford gates, controls and square roots can be combined to climb to the fourth level, mimicking the behavior of diagonal gates. 

Based on the work presented here and experimental data (for a small number of qubits) we are cautiously optimistic that our approach will lead to a complete understanding of the fourth level of the Clifford hierarchy as well as shed bright light into the entire hierarchy.
We lay out below some interesting and promising future directions.

\textbf{Lifting the Clifford group.} While we completely characterized Hemitian Clifford gates that climb to the third level, it would be interesting (and important!) to determine if and when the square root of a Hermitian Clifford gate is again in the hierarchy.
As discussed, basic gates such as the Hadamard gate (see also Example~\ref{exa-hadamard}) are not. 
We point out in here that preliminary data shows that square roots of Hermitian Clifford gates belong to the fourth level (see also Example~\ref{nonexa}).
Beyond the Hermitian case, we saw in Examples~\ref{exa-sqrtT} and \ref{exa-sqrtT1} that non-Hermitian Clifford gates behave entirely differently and require a separate treatment.

\textbf{Lifting the third level.} 
The characterization given in Section~\ref{Sec4} completely relied on the symplectic representation of the Clifford group. Lifting third level gates will likely require different techniques.
In Section~\ref{Sec5}, we give some preliminary results in this direction. It would be interesting to explore if and when those results can be extended further.
Here, there are at least two possible ways to branch out.
First, establish conditions for a third level permutation gate to climb to the fourth level.
Examples~\ref{exa-toff} and~\ref{exa-swap2} again shed some light on this. Towards this direction, it would be crucial to use the classification of third level permutations~\cite{HRT25} to obtain sufficient conditions analogous to Theorem~\ref{T-gen-climb}. 
Second, Theorems~\ref{T-C4} and~\ref{T-C4-1} highlight third level gates whose controlled versions climb to the fourth level. We plan to explore whether or not other third level gates have this property.

\textbf{Considering $2^m$th roots.} The square roots of Hermitian gates considered in this paper fall under a much more broader category of gates. Namely, let $U$ be a Hermitian gate and denote
\begin{equation}
{\sf sq}_m(U) = {\sf exp}\left(i\frac{\pi}{2^{m+1}}U\right) = {\sf cos}\left(\frac{\pi}{2^{m+1}}\right)I + i{\sf sin}\left(\frac{\pi}{2^{m+1}}\right)U.
\end{equation}
With this notation we have $\widehat{U} = {\sf sq}_1(U)$, which naturally leads to the following.
\begin{prob}
Let $U\in \cC^{(k)}$ be a Hermitian gate. Under what conditions does ${\sf sq}_m(U)\in \cC^{(k+m)}$ hold?
\end{prob}
In Example~\ref{exa-sqrtT}, we argued that if $E$ is a Hermitian Pauli gate then ${\sf sq}_m(E)\in \cC^{(m+1)}$.
Preliminary computations show that similar results hold for other Hermitian gates.
For instance, we have that ${\sf sq}_m({\sf C}^{(k)}(X)) \in \cC^{(k+m+1)}$ (recall that ${\sf C}^{(k)}(X) \in \cC^{(k+1)})$ and for $U$ being a controlled-SWAP gate we have that ${\sf sq}_m(U) \in \cC^{(m+3)}$ (recall that control-SWAP is a third level gate).
The ultimate goal would be to obtain an analogous result with Theorem~\ref{T-diagonal} that holds for the diagonal part of the hierarchy.
\begin{prob}
Let $U\in \cC^{(k)}$ be a Hermitian gate. Under what conditions does ${\sf sq}_m({\sf C}^{(n)}(U))$, $\sfC^{(n)}({\sf sq}_m(U))\in \cC^{(k+n+m)}$ hold?
\end{prob}
\section*{Author contributions}
All Authors contributed equally to this work.
\section*{Acknowledgment}
Most of this work was performed during the Research
Experiences for Undergraduates (REU-Site) Program
“Cryptography and Coding Theory at the University of
South Florida,” which ran during Summer 2025
(\href{usf-crypto.org/reu-program/}{usf-crypto.org/reu-program/}). We would like to thank Olav Tirkkonen for helpful discussions.

\bibliography{refs}
\bibliographystyle{abbrvurl}


\appendix
\section{On the number of Hermitian Clifford gates that climb to the third level}
In this appendix, we give the number of diagonal, permutation, and general Hermitian Clifford gates that climb to the third level. In particular, the results show that number of such gates is exponentially small relative to total number of Clifford gates.
Here, all the counting is done up to a Pauli gate.
\begin{lem}
There are $(2^n-1)(2^n-2)/6$ diagonal Clifford gates $C \in \cC_n^{(2)}$ such that $\widehat{C}\in \cC_n^{(3)}$.
\end{lem}
\begin{proof}
We need to count the number of symplectic matrices $F_C$ as in~\eqref{e-sym-symp} such that $\dim(\Res(I+F_C))=2$.
Recall that, in~\eqref{e-sym-symp}, $A$ must be symmetric and have zero diagonal (so that $C$ is Hermitian). 
Thus, this is equivalent to counting the number of $n\times n$ binary symmetric matrices $A$ that have zero diagonal and $\rk(A) = 2$.
It is readily verifiable that these are precisely the matrices of the form $A = \ba^\top\bb + \bb^\top\ba$ where $\ba\neq \bb$ are two nonzero vectors in $\F_2^n$.
This gives $(2^n-1)(2^n-2)$ options.
However, swapping $\ba$ and $\bb$ and replacing either $\ba$ or $\bb$ with $\ba+\bb$ yields the same matrix, and thus we were overcounting by a factor of 6.
\end{proof}
\begin{lem}
There are $(2^n-1)(2^{n-1}-1)$ permutation Clifford gates $C \in \cC_n^{(2)}$ such that $\widehat{C}\in \cC_n^{(3)}$.
\end{lem}
\begin{proof}
We need to count the number of symplectic matrices $F_C$ as in~\eqref{e-perm-symp} such that $\dim(\Res(I+F_C))=2$.
Recall that, in~\eqref{e-perm-symp}, $B$ must be invertible and satisfy $B = B^{-1}$.
This is equivalent to counting the number of invertible matrices (along with $B^{-1} = B$, although as we will see, this is not relevant) $B$ such that $\rk(I+B)=1$.
Thus, we are looking for all invertible matrices $B$ such that $I + B = \ba^\top\bb$ where $\ba,\bb\in \F_2^n$ are nonzero.
Since $B$ is invertible and by the Determinant Lemma, we have
\begin{equation*}
1 = {\sf det}(B) = {\sf det}(I+\ba^\top\bb) = 1 + \ba\bb^\top,
\end{equation*}
which imposes the condition $\ba\bb^\top = 0$. 
Thus, we have $2^n-1$ options for one vector, and once that is fixed, there are $2^{n-1}-1$ options (that is, number of orthogonal nonzero vectors with a given vector) for the second one.
\end{proof}
\begin{lem}
There are $(2^{2n}-1)(2^{2n-2}-1)/3$ Clifford gates $C\in \cC_n^{(2)}$ such that $\widehat{C}\in \cC_n^{(3)}$.
\end{lem}
\begin{proof}
We need to count the number of symplectic matrices $F$ such that $F^2 = I$, $\dim(\Res(F)) = 2$, and $\inners{\bv}{\bv F}=0$ for all $\bv\in \F_2^{2n}$.
By Lemma~\ref{L-com}, we have that $\inners{\bx}{\by} = 0$ for all $\bx,\by\in \Res(F)$.
Thus, we need to count subspaces of $\F_2^{2n}$ of dimension 2 that are self-orthogonal with respect to the symplectic inner product.
For this, there are $2^{2n}-1$ choices for the first basis vector. Then, there are $2^{2n-1}$ orthogonal vectors with the first one. But we have to exclude the zero vector and the first fixed vector.
As before, to avoid overcounting, we have to divide by 6, obtaining $(2^{2n}-1)(2^{2n-2}-1)/3$.
The proof is concluded with the observation that for every such subspace (spanned by two vectors $\bv_1,\bv_2$) there is exactly one $F$ that satisfies $\inners{\bv}{\bv F}=0$ for all $\bv$, namely, $F = F_{\bv_1}F_{\bv_2}$.
\end{proof}

\end{document}